\documentclass[10pt,a4paper]{article}
\usepackage{xcolor}
\usepackage{amsthm}
\usepackage{mathtools}
\usepackage{enumitem}
\usepackage{amsthm}
\usepackage{graphicx}
\usepackage{tikz}
\usepackage{amsmath}

\usetikzlibrary{arrows}
\usetikzlibrary{intersections}

\setlist[itemize]{leftmargin=*}

\newtheorem{theorem}{Theorem}[section]
\newtheorem{corollary}{Corollary}[theorem]
\newtheorem{lemma}[theorem]{Lemma}
\theoremstyle{definition}
\newtheorem{definition}{Definition}[section]

\begin{document}
\begin{figure}[t]
\includegraphics[scale = 1.5]{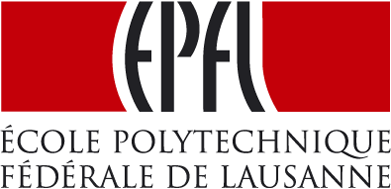}
\centering
\end{figure}

\title{\textbf{Achieving Competitiveness in Online Problems}}
\author{Mustafa Safa Ozdayi \\ Supervised by: Ola Svensson and Sangxia Huang}
\date{13 January 2017}
\maketitle

\begin{abstract}
In the setting of online algorithms, the input is initially not present but rather arrive one-by-one over time and after each input, the algorithm has to make a decision.
Depending on the formulation of the problem, the algorithm might be allowed to change its previous decisions or not at a later time.
We analyze two problems to show that it is possible for an online algorithm to become more competitive by changing its former decisions.
We first consider the online edge orientation in which the edges arrive one-by-one to an empty graph and the aim is to orient them in a way
such that the maximum in-degree is minimized.
We then consider the online bipartite $b$-matching.
In this problem, we are given a bipartite graph where one side of the graph is initially present and where the other side arrive online.
The goal is to maintain a matching set such that the maximum degree in the set is minimized.
For both of the problems, the best achievable competitive ratio is $\Theta(\log n)$ over $n$ input arrivals when decisions are irreversible.
We study three algorithms for these problems, two for the former and one for the latter, that achieve $O(1)$ competitive ratio by changing $O(n)$ of their decisions over
$n$ arrivals.
In addition to that, we analyze one of the algorithms, the shortest path algorithm, against an adversary.
Through that, we prove some new results about algorithms performance.
\end{abstract}
\newpage

\section{Introduction}
\label{sec:introduction}
In this paper, we study the online algorithms.
Online algorithms get their input one-by-one during their computation unlike their offline counterparts which are given the whole input before their computation.
\textit{Competitive analysis} is concerned with how an online algorithm performs with respect to its offline counterpart. \textit{Competitive ratio} is the ratio between the
performance of an online algorithm and its offline counterpart. An online algorithm is called \textit{competitive} if its competitive ratio is bounded by some constant. More formally,
we call an online algorithm $c \textit{-competitive}$ if there exists an arbitrary constant $a$ such that for all finite input sequences $I$, we have $$ C(ALG(I)) \leq c \cdot C(OPT(I)) + a, $$ where $C(OPT(\cdot))$ is the cost
of the optimal offline algorithm and $C(ALG(\cdot))$ is the cost of the online algorithm. We note that an optimal offline algorithm knows the
entire sequence in advance and can process it with minimum possible cost.

An important concept about online algorithms is \textit{recourse}, i.e. being able to change past decisions.
Indeed, not being able to see future inputs can be problematic for an online algorithm. A single new input might be inconsistent with the solution that the algorithm has built until that point and can increase
the competitive ratio greatly. The main idea is, by allowing the algorithm to change some of its past decision, to make the solution more compatible with arriving inputs and thus, to lower the competitive ratio.

Certainly, for an online algorithm there is a tradeoff between the number of allowed recourse operations and its competitive ratio. If an arbitrary number of recourse operations are allowed, then the online algorithm can
simply simulate the offline optimal algorithm and achieve $1$-competitiveness. All it has to do is to make some arbitrary decisions as the inputs arrive and to run the offline optimal algorithm
once it has all the inputs. Due to that, we want to keep the number of recourse operations bounded from above to maintain the online nature of online algorithms.

The rest of the paper is organized as follows: in Section 2, we present the online edge orientation and two different algorithms, the shortest path algorithm
and the all-flip algorithm, that achieve constant competitive ratio under certain assumptions. In Section 3, we establish some results regarding the shortest path algorithm through analyzing
it against an adversary. In Section 4, we present the online bipartite $b$-matching and an algorithm that achieves a constant competitive ratio for it.
In the last section, we give some concluding remarks.
\newpage

\section{Online Edge Orientation}
In the online edge orientation, we are given a set of nodes and a set of undirected edges. The edges are not initially present but rather, they are added one-by-one.
When an edge is added, we must orient it towards one of its endpoints. The objective is to orient the edges in a way such that the maximum in-degree of the graph is minimized after $n$ edge additions.
More precisely, the cost of the problem is defined as the maximum in-degree in the graph after $n$ arrivals.
Note that $n$ is the number of edges in the graph after $n$ edge additions and  there is no upper bound on the number of nodes.
For this problem, a recourse operation is simply reorienting an edge (which we also call as flip sometimes) after it was given an initial orientation at its addition.

The first algorithm we present for this problem assumes the set of arriving edges are acyclic. Before presenting the algorithm, we present a lemma that gives a tight bound on the competitive
ratio of any algorithm under that assumption when we do not allow recourse.

\begin{lemma}\label{lwbnd}
If the given edge set is acyclic and if the reorientation of the edges are not allowed, the best competitive ratio that can be achieved for the online edge orientation is $\Theta(\log n)$.
\end{lemma}
\begin{proof}
We first show the lower bound by constructing a sequence of the edges. First assume that $n \neq 1$ and divisions below yield an integer value.
We keep track of the maximum in-degree round by round. In the first round, we put n/2 edges (that arrive one-by-one) such that none of them are incident.
This will create n/2 components that consist of an edge and two nodes in which one has in-degree 0 and the other has in-degree 1.
In the second round, we put n/4 edges between the nodes that have in-degree $1$. After this round we have n/4 components that
consist of two edges and four nodes such that two has in-degree 0, one has in-degree 1 and one has in-degree 2.
If we continue to keep putting edges in this fashion, we see that maximum in-degree of
the graph increases by 1 after each round. The number of rounds is simply the number of times we can halve $n$ which is $\log_2 n$.
If dividing $n$ by powers of $2$ does not yield an integer, we take ceil of the divisions if $n = 2^k - 1$ for some
integer $k \geq 2$ and we take floor of them otherwise. Then, the number of rounds and consequently
the maximum in-degree is $\lceil \log_2 n \rceil$ if $n = 2^k - 1$ or $\lfloor \log_2 n \rfloor$ otherwise. Finally for $n = 1$, the
maximum in-degree is simply $1$.

For the upper bound, consider the simple algorithm that orients edges towards the node with smaller in-degree.
Now observe that forcing the maximum in-degree to increase by 1 requires connecting two nodes that has maximum in-degree.
For any other case, the algorithm orients the edge towards the node with smaller in-degree and avoids increasing the maximum in-degree.
Then, it is easy to prove any sequence that achieves $k$ in-degree requires at least $2^k - 1$ edges by induction.
Base case $k = 1$ trivially holds. Assume any sequence that achieves $k-1$ in-degree requires at least $2^{k-1} -1$ edges. Then we can construct a sequence
to achieve $k$ in-degree with two $k-1$ achieving sequences and a single edge. Note that $k-1$ achieving sequences must be defined on distinct set of
nodes in order to ensure to have two distinct nodes with $k-1$ in-degree.
Thus, we need at least $2\cdot(2^{k-1} -1) + 1 = 2^k -1$ edges.
Of this we conclude that any sequence that achieves $\log_2 n$ in-degree requires at least $2^{\log_2 n} - 1 = n-1$ edges. Since the last remaining
edge can not possibly increase the maximum in-degree by itself, this is the upper bound.

Finally, remember that our edge sequence is acyclic. This means they form a forest.
Further, it implies there exists an orientation of the edges such that the maximum in-degree is 1. One can find such orientation by arbitrarily picking a node as the root and orienting all the edges
away from it in every tree in the forest. Hence, the competitive ratio is $\Theta (\log n)$.
\end{proof}

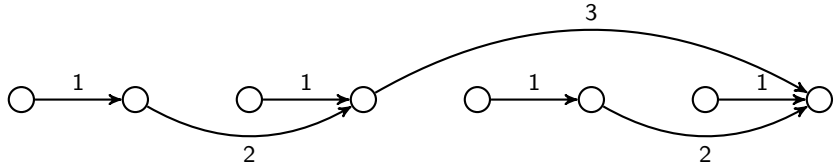
\begin{figure}[ht] \label{fig:lowerBound}
  \centering
\begin{tikzpicture}[->,>=stealth',auto,node distance=1.5cm,
  thick,main node/.style={circle,draw,font=\sffamily\bfseries}]

  \node[main node] (1) {};
  \node[main node] (2) [right of=1] {};
  \node[main node] (3) [right of=2] {};
  \node[main node] (4) [right of=3] {};
  \node[main node] (5) [right of=4] {};
  \node[main node] (6) [right of=5] {};
  \node[main node] (7) [right of=6] {};
  \node[main node] (8) [right of=7] {};

  \path[every node/.style={font=\sffamily\small}]
    (1) edge node [right][above] {1} (2)
    (3) edge node [right][above] {1} (4)
    (5) edge node [right][above] {1} (6)
    (7) edge node [right][above] {1} (8)
    (2) edge[bend right] node [left][below] {2} (4)
    (6) edge[bend right] node [left][below] {2} (8)
    (4) edge[bend left] node [left][above] {3} (8);
\end{tikzpicture}
\caption{Construction for the lower bound with $n = 8$ where the last edge is redundant.  Numbers on the edges indicate their addition round. As can be seen, maximum in-degree is $\log_2 8 = 3$.}
\end{figure}

We have showed the competitive ratio for this problem is $\Theta(\log n)$ when there are no reorientations.
We now present an algorithm that achieves $O(1)$ competitive ratio by doing $O(n)$ reorientations over $n$ arrivals.

\subsection{Shortest Path Algorithm}
We remind that we assume the set of arriving edges are acyclic. The algorithm given in \cite{Gupta:2014:MAO:2634074.2634109} works in a greedy
fashion to achieve $O(1)$ competitiveness under that assumption. Basically, the algorithm maintains a set of nodes such that no node has more than $c$ in-degree
over $n$ edge additions for some constant $c \geq 2$ which is referred as the in-degree constraint of the graph. Hence, it is $c$-competitive.
We note that achieving $1$-competitiveness is possible under the given edge set however,
it is known that there exists instances which might require any algorithm to do $\Omega(n \log n)$ flips in that case (\cite{Grove:1995:OPM:645930.758390}).
So by relaxing the constraint, authors try to approximate the optimal solution within a constant degree by doing fewer number of flips.
For the rest of the section we let the in-degree constraint be $c=2$ and call a node \textit{saturated} if it has 2 in-degree.
We now present a simple lemma which is rather an observation about our graph.

\begin{lemma}
Let $u$ be a node in the forest. We denote the tree that $u$ belongs to by $T_u$ and the number of nodes it has by $|T_u|$. Then, there exists a node $u'$ in $T_u$ such that it is unsaturated and there is a path
from $u'$ to $u$, i.e. $P_{u' \rightarrow u}$. We denote the length of that path by $|P_{u' \rightarrow u}|$.
\end{lemma}
\begin{proof}
Let $|T_u| = k$. Since $T_u$ is a tree, it has $k-1$ edges. To saturate all the nodes in $T_u$ we need $2k$ edges. Hence, there must be some unsaturated nodes in it.
If $u$ is not saturated, we simply have $u = u'$. If $u$ is saturated then it has $2$ edges coming to it.
By doing a tree traversal that starts at $u$, we will eventually arrive to an unsatured node as we know all the nodes
in $T_u$ can not be saturated. The node we arrive is $u'$ and we can extract the path $P_{u' \rightarrow u}$ from the traversal.
\end{proof}

As its name indicates, the algorithm handles constraint violations by flipping the edges towards the shortest of the available paths. Now
suppose some edge $(u, v)$ arrives. If both endpoints are unsaturated, the algorithm picks one randomly.
If one of them is not saturated, then it orients the edge to it. Now assume both $u$ and $v$ are saturated. Since we assumed that the set of arriving edges are acyclic, $T_u \neq T_v .$
By the lemma above, there are unsaturated nodes in $T_u$ that have a path to $u$. Let $u'$ be any of the closest among such nodes by $|P_{u' \rightarrow u}|$. Similarly let $v'$ be such node in $T_v$. It is clear that flipping
process must end in an unsaturated node. This implies the minimum number of flips we can have is the minimum of $|P_{u' u}|$ and $|P_{v'v}|$.
So the algorithm finds $u'$ and $v'$, computes $|P_{u' \rightarrow u}|$ and $|P_{v' \rightarrow v}|$, picks the path with distance $\min(|P_{u' \rightarrow u}|, |P_{v' \rightarrow v}|)$
and flips all the edges on that path until it reaches an unsaturated node which is either $u'$ or $v'$.
Then, the number of flips caused by the addition of $(u, v)$ is simply given by $\min(|P_{u' \rightarrow u}|, |P_{v' \rightarrow v}|).$

We now give an upper bound to the number of flips that $(u, v)$ can cause.
Suppose the algorithm oriented the towards $u$ and consequently, flipped all the edges on $P_{u' \rightarrow u}$.
We observe that every node on that path must be saturated except $u'$. This is simply from the definition of $u'$, i.e.  $u'$ is the closest unsaturated node to $u$.
So $u$ has two edges pointing to it, the nodes that are at the other endpoints of those two edges must have two edges pointing to them too and so on
until $u'$. Essentially, the tree $T_u$ contains a complete binary tree at least up to depth $|P_{u' \rightarrow u}| - 1$. So $|P_{u' \rightarrow u}| \leq \log_2  |T_u| $
and therefore there can be at most $\log_2 |T_u|$ flips. Consequently, we can upper bound the number of flips that any edge addition can cause by $\log_2 n$.

As last, we bound the total number of flips over $n$ edge additions. From the result above, a straightforward upper bound is $O(n \log n)$. However, by using the fact that the algorithm always
flips the shortest of the available paths, it is possible to obtain a tighter bound. The worst case is simply the case where the edges
are stacked in just two trees over $n$ additions. We write the recurrence accordingly and present its solution.

\begin{lemma}
Let $T(k)$ be an upper bound on the number of flips after $k$ edge additions. Then,
$$ T(n) \leq \max_{1 \leq n_1 < n}\{T(n_1) + T(n-n_1) + \log_2 \text{min}(n_1, n - n_1)\} $$ with base case $T(1) = 0$ solves to $T(n) \leq n - \log_2 n - 1 .$
\end{lemma}
\begin{proof}
Base case trivially holds. Now assume $T(k) \leq k - \log_2 k - 1$ for $k < n$. We will show
$T(n) \leq n - \log_2 n - 1$.\\

1) If $n_1 > n - n_1$:
\begin{align*}
T(n) & \leq n_1 - \log_2 n_1 -1 + (n-n_1) - \log_2 (n-n_1) - 1 + \log_2 (n-n_1) \\
& = n - (\log_2 n_1 + 1 ) - 1 \\
& = n - \log_2{2n_1} - 1 \\
& < n - \log_2 n - 1 \text{ as } 2n_1 > n.
\end{align*}

2) If $n_1 \leq n - n_1$:
\begin{align*}
T(n) & \leq n_1 - \log_2 n_1 -1 + (n-n_1) - \log_2 (n-n_1) - 1 + \log_2 (n_1) \\
& = n - (\log_2(n- n_1) + 1) - 1 \\
& = n - \log_2{2(n- n_1)} - 1 \\
& \leq n - \log_2 n - 1 \text{ as } n \geq 2n_1.
\end{align*}
\end{proof}
Note that we improved the analysis of the original authors. We restate the main theorem presented by them accordingly.
\begin{theorem}\label{shortest_mainTheo}
Given the edge set is acyclic, the shortest path algorithm maintains at most $2$ in-degree in all nodes by doing at most
$n - \log_2 n - 1$ edge reorientations over $n$ edge additions. Consequently, the shortest path algorithm achieves $O(1)$ competitiveness with $O(n)$ edge reorientations over $n$ arrivals.
\end{theorem}
We finally note that generalizing results for an arbitrary in-degree constraint $c \geq 2$ is possible by changing the base of the logarithm.
As we have showed, we have binary trees up to unsaturated nodes when $c=2$. For an arbitrary $c \geq 2$, we would have $c$-ary trees.
We can upper bound the maximum number of flips per step as $\log_c n$ and upper bound the total number of flips over $n$ edge additions as $\cfrac{n - log_2 n - 1}{log_2 c}$.

\subsection{All Flip Algorithm}
We now present the algorithm given in \cite{Brodal99dynamicrepresentations} which achieves $O(1)$ competitive ratio with $O(n)$ reorientations when the arboricity of the graph is bounded by $c$ during the entire edge addition sequence.
Like the previous algorithm, algorithm keeps the in-degree of all nodes below some constant $c$ over $n$ edge additions.
The arboricity $c$ of a graph is defined as $$ c =  \max_{J} \frac{|E(J)|}{|V(J)| - 1}, $$ where $J$ is any subgraph in $G = (V,E)$ with $|V(J)| \geq 2$ nodes and $|E(J)|$ edges.
The importance of the bound on arboricity is due to the following theorem.

\begin{theorem} (Nash-Williams~\cite{Nash-Williams01011964})
A graph $G = (V, E)$ has arboricity $c$ if and only if $c$ is the smallest number of sets $E_1, \dots, E_c$ that $E$ can be partitioned into, such that each subgrah $(V, E_i)$ is a forest.
\end{theorem}
In other words, if the arboricity of $G$ is bounded by $c$ then we can partition $G$ into $c$ forests. In each of these forests, we can pick an arbitrary node as the root and orient all the edges away from it.
So, in-degree of each node can be at most one in each forest and at most $c$ in the entire graph as a node can be in all of the forests.
Thus for any graph $G$ whose arboricity is bounded by $c$, there exists an orientation of the edges such that no node has more than $c$ in-degree.

The set of edge orientations such that no node has more than $\delta$ in-degree is called $\delta$-orientations.
Assuming the arboricity of the graph is bounded by $c$ during the entire edge addition sequence, the algorithm given by them finds a $\Delta$-orientation if the given edge set allows a $\delta$-orientation for $\delta \geq c$ and if $\Delta \geq 2\delta$.
It runs as follows: when the algorithm gets a new edge $(u, v)$, it orients it arbitrarily to one of its endpoints.
Assume it is oriented towards $v$. Then, if the in-degree of $v$ is still bounded by $\Delta$, algorithm proceeds to the next input.
Otherwise ($v$ has $\Delta + 1$ in-degree), the algorithm flips all the incoming edges at $v$. If this causes
some other nodes to have $\Delta + 1$ in-degree,  algorithm flips their incoming edges too. This flipping process continues until all the nodes have at most $\Delta$ in-degree.

Before presenting the proof that shows the algorithm indeed works, i.e. the flipping process eventually terminates in a state where all nodes have at most $\Delta$ in-degree,
we give an example case to build some intuition.
Assume the set of arriving edges form a tree. So, our input allows a $\delta$-orientation for $\delta = 1$ and say we want a $\Delta$-orientation for
$\Delta = 2$.
Define potential $\Psi(t)$ as the number of edges that have a different orientation with respect to a fixed $\delta$-orientation in our graph after $t$ edge additions and note that this potential
can increase at most by 1 with a new addition.
Assume after $t$'th edge addition a node $v$ gets 3 incoming edges.
We know at most 1 of those edges can have the same orientation with respect to our fixed $\delta$-orientation.
Hence, at least 2 of these incoming edges must have a different orientation.
We see that that if we flip all the $3$ incoming edges at $v$, $\Psi(t)$ must decrease at least by 1. As noted before, flipping those edges
might cause problems at some other nodes. However, by the same argument we can keep reducing $\Psi(t)$ by repeating the process at them. The following
lemma gives an upper bound on the number of flips over $n$ edge addition and proves that the algorithm eventually achieves the desired state.

\begin{lemma}
For an arboricity c preserving edge addition sequence $\sigma$ on an empty graph, let $G_t = (V, E_t)$ denote the graph after $t$'th edge addition and let length of the sequence be $n$.
If the given sequence $\sigma$ allows a $\delta$-orientation, then the algorithm does at most $$ n\frac{\Delta+1}{\Delta + 1 - 2\delta} $$ edge reorientations on $\sigma$ given that $\Delta \geq 2\delta$ and $\delta \geq c$.
\end{lemma}
\begin{proof}
Consider a fixed $\delta$-orientation built by $\sigma$ on the node set $V$. Let $\bar{E_t}$ denote the edge set of the fixed $\delta$-orientitation after its $t$'th edge addition.
An edge in $E_t$ is denoted \textit{good} if it has the same orientation in $\bar{E_t}$. Otherwise it is denoted \textit{bad}. Define a potential function as the following
$$ \Psi(t) =  \text{the number of bad edges in } E_t .$$

Note that this potential is non-negative and $\Psi(0) = 0$.
Assume after $t$'th addition, a node $v$ gets $\Delta + 1$ incoming edges.
We note that at most $\delta$ of them can be good and consequently, at least $\Delta + 1 - \delta$ of them are bad.
Thus, if all the incoming edges are flipped at $v$, at most $\delta$ edges may become bad and at least $\Delta + 1 - \delta$ edges become good.
So by flipping all the incoming edges at $v$, we reduce the  bad edges by at least $(\Delta + 1 - \delta) -\delta = \Delta + 1 - 2\delta$.
After $t$'th addition, the number of bad edges can be at most $\Psi(t-1) + 1$  and we can make at most $\Psi(t)$ of them good (simply from the definition, i.e.
$\Psi(t)$ is the number of bad edges after $t$'th addition) in that step. Since we know an all-flip (flipping all the incoming edges at some node) decreases bad edges by
at least $\Delta + 1 - 2\delta$, we can have at most $(\Psi(t-1) + 1 - \Psi(t))/\Delta + 1 - 2\delta$ all-flips after $t$'th addition.
Summing this over $n$ edge additions gives us an upper bound on the total number of all-flips.
We have
$$\sum_{t=1}^n \frac{\Psi(t-1) + 1 -\Psi(t)}{\Delta + 1 - 2\delta} = \frac{\Psi(0) + n -\Psi(n)}{\Delta + 1 - 2\delta} \leq  \frac{n}{\Delta + 1 - 2\delta},$$
and since an all-flip is simply flipping $\Delta + 1$ edges, the total number of flips are bounded from above by
$$  n\frac{\Delta+1}{\Delta + 1 - 2\delta} .$$
\end{proof}
For our example case, we had $\delta = 1$ and $\Delta = 2$. Hence, the algorithm would do at most $3n$ flips over $n$ additions.

\section{Adversary Against the Shortest Path Algorithm}
We now prove some lower bounds for the shortest path algorithm. We do this
by analyzing the algorithm against an adversary who supplies edges that forces the algorithm to make
expensive choices. We again assume the in-degree constraint is $2$ and the set of arriving edges are acyclic.
We remind that a node is called saturated if it has $2$ in-degree and start by giving a few definitions.

\begin{definition} \label{path-def}
{\textit{Variant of the shortest path algorithm.}}
An algorithm is called a variant of the shortest path algorithm if it satisfies the following conditions.
\begin{enumerate}
\item It reorients edges only when it is necessary (i.e. when there is a constraint violation)
and always towards the shortest of the available paths. It can arbitrarily chose among such paths when there is a tie.
\item It can orient an edge between two unsaturated nodes in any manner.
\end{enumerate}
\end{definition}
The shortest path algorithm we presented in the previous section of course fits this description.

\begin{definition}
{\textit{Adversary}.} An adversary is defined as a player who has complete knowledge about the algorithm and can pass any valid input to it.
\end{definition}

\begin{definition}{\textit{$T^m.$}}
$T^m$ is the set of trees with saturated root and shortest path of length $m$ where $m \geq 1$.
Further for a $t_m \in T_m$, we define its size as the number of edges it has and denote it by $|t_m|$.
\end{definition}

We note that we can express the state of our graph at each step as a combination of $t_m$ trees.
We now show how an adversary can force the construction of such trees. Basically by using them, an adversary would be
able to force more complex constructions.
\newpage

\begin{lemma} \label{sizeLemma}
An adversary can force any variant of the shortest path algorithm to construct a $t_m \in T_m$ by using at most $5 \cdot 2^{m-1} - 2$ edges.
\end{lemma}
\begin{proof}
We give a proof by induction. Consider the base case $m = 1$. Adversary needs at most $4$ single nodes. Let them be $a, b, c$ and $d$. Adversary first gives the edge $(a, b)$. By point 2 of the
Definition \ref{path-def}, this edge can be oriented to any direction. As everything is symmetric, we can assume it is oriented towards $b$ without loss of generality. Adversary then gives the edge $(c, d)$. By the same argument,
assume it is oriented towards $d$ without loss of generality. Then by giving the edge $(b, d)$, adversary will have a $t_1$ as either orientation will cause a saturated
node with shortest path of length $1$. Following figure illustrates this process.

\begin{figure}[ht] \label{figt1}
  \centering
\begin{tikzpicture}[->,>=stealth',auto,node distance=3cm,
  thick,main node/.style={circle,draw,font=\sffamily\bfseries}]

  \node[main node] (1) {a};
  \node[main node] (2) [right of=1] {b};
  \node[main node] (3) [right of=2] {c};
  \node[main node] (4) [right of=3] {d};

  \path[every node/.style={font=\sffamily\small}]
    (1) edge node [right][above] {1} (2)
    (3) edge node [right][above] {2} (4)
    (2) edge[bend right] node [left][below] {3} (4);
\end{tikzpicture}
\caption{Construction of a $t_1$ with root $d$. Numbers on the edges indicate their addition order.}
\end{figure}
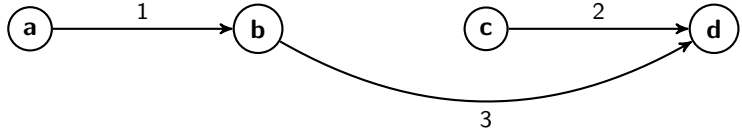

Assume the construction of $t_{m-1}$ can be forced by using at most $5 \cdot 2^{m-2} - 2$ edges.
To force the construction of a $t_m$, adversary first forces the construction of two $t_{m-1}$ using disjoint sets of nodes. Let $a$ be the root of the first $t_{m-1}$, let $b$ be the root
of the second and let $c$ be a single node separate from the constructed trees. Adversary first gives the edge $(a, c)$. By point 1 of the Definition \ref{path-def}, the edge will orient towards $c$.
Adversary then gives the edge $(b, c)$ and by the same argument, it will be oriented towards $c$. After these two edges, $c$ is a saturated
node with shortest path of length $m$ giving us a $t_m$ with root $c$.
Hence, the total number of edges used to force the construction of a $t_{m}$ is at most $2|t_{m-1}| + 2 = 5 \cdot 2^{m-1} - 2$.
\end{proof}

\begin{corollary}
An adversary can force any variant of the shortest path algorithm to flip $\log_2 m$ edges in a single step using at most $5m-3$ edges in total.
\end{corollary}
\begin{proof}
Adversary first forces the construction of two $t_{\log_2 m}$ using Lemma \ref{sizeLemma}. Then by giving the edge that connects the roots of these trees,
adversary will cause $\log_2 m$ flips. In total, $2|t_{\log_2 m}| + 1 = 5m-3$ edges are used.
\end{proof}
Notice that if we set the total number of edges as $n = 5m-3$ then we have $\Omega(\log n)$ flips.
Recall that we have the upper bound $O(\log n)$ for the number of edge flips in a single step
so, we conclude that this bound is tight.

\begin{lemma}
Any variant of the shortest path algorithm can be forced to flip $\Omega(n)$ edges in total after $n$ edge additions.
\end{lemma}
\begin{proof}
The adversary first constructs $k+1$ disjoint $t_m$ trees using Lemma \ref{sizeLemma}.
\begin{figure}[ht]
  \centering
  \begin{tikzpicture}[{root/.style={circle,draw,fill,inner sep=2 pt}}]
  \node[root] (rootA) {};
  \node[root,xshift=3cm] (rootB) {};
  \node[root,yshift =-5cm,xshift=3cm] (rootC) {};
  \draw (rootA) -- +(-1,-2) -- node[yshift=0.75cm]{$t_m$} +(1,-2) -- (rootA);
  \draw (rootB) -- +(-1,-2) -- node[yshift=0.75cm]{$t_m$} +(1,-2) -- (rootB);
  \draw [dotted] (3, -2.5) -- (3, -4.5);
  \draw (rootC) -- +(-1,-2) -- node[yshift=0.75cm]{$t_m$} +(1,-2) -- (rootC);
\end{tikzpicture}
\caption{A single $t_m$ on the left, k $t_m$ on the right.}
\end{figure}
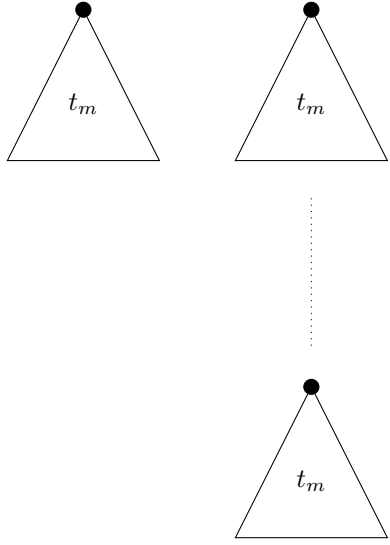
Now when we look at the figure below, we see that for every $t_m$ on the right,
an edge can be added between its root and the root of $t_m$ on the left. For every such edge, there will be $m$ flips due to
point 1 of Definition \ref{path-def}. Thus, forcing the algorithm to make $km$ flips requires $k$ edges after reaching the construction in the figure.
Total number of edges required to do $km$ flips is then given by $$f(k, m) = k + (k+1)|t_m| \leq n .$$  Clearly, $f(k, m)$ is maximized when $m=1$ as $|t_m|$ grows exponentially with $m$.
Hence, an upper bound on $k$ when $m=1$ will be an upper bound on the total number of flips for this construction. Solving the inequality for $k$ with $m = 1$ gives $$k \leq \frac{n-|T_1|}{|T_1| + 1} = \frac{n-3}{4}.$$
So, we have $f(k, m) \leq \frac{n-3}{4}$ with equality iff $m = 1$ and $\frac{n-3}{4}$ is an integer.
\end{proof}
By combining this result with Theorem \ref{shortest_mainTheo}, we see that the total number of flips is $\Theta(n)$. So, we can not asymptotically do better by handling edges between unsaturated nodes cleverly.

\begin{lemma}
A single edge can be forced to flip  $\Omega(\log n)$ after $n$ edge additions.
\end{lemma}
\begin{proof}
Consider the initial setting where the adversary puts an edge between single nodes $a$ and $b$. Without loss of generality, assume the algorithm directed it towards $b$.
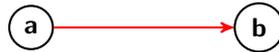
\begin{figure}[ht] \label{mrkd}
  \centering
\begin{tikzpicture}[->,>=stealth',auto,node distance=3cm,
  thick,main node/.style={circle,draw,font=\sffamily\bfseries}]

  \node[main node] (1) {a};
  \node[main node] (2) [right of=1] {b};

  \path[every node/.style={font=\sffamily\small}]
    (1) edge[red] node [right][above] {} (2);
\end{tikzpicture}
\caption{Initial setting. The edge we count flips for is in red.}
\end{figure}

To flip the red edge back to $a$, adversary connects 2 $t_1$ to $b$ due to point 1 of Definition \ref{path-def}.
The first one will saturate $b$ and the second one will force the flip in the
direction of $a$ and consequently will flip the red edge. Note that we now have a $t_2$ rooted at node $b$.
To flip the red edge again, adversary can connect 2 $t_3$ to $a$. By the same reasoning, this will flip the edge towards $b$. Adversary can continue to flip the red edge in this fashion.

The total number of edges required to do $k$ flips on the red edge is then given by $$f(k) = 1 + \sum_{m=1}^{k} 2|t_{2m -1}| + 2.$$
By using Lemma \ref{sizeLemma}, it can be rewritten as $$f(k) = -ax + b\exp\{cx\} + d$$ for some positive constants $a, b, c, d$ and for Euler constant $\exp$.
Finally, solving $f(k) \leq n$ for $k$ will give $k = \Omega(\log n).$
\end{proof}

As last, we consider a slightly different lower bound. Namely, we wonder whether a shortest path algorithm can maintain in-degree constraint
over $n$ additions by doing constant number of flips at each step. As we have showed, an adversary can force any number of flips up to $O(\log n)$ at a step
if the shortest path algorithm satisfies Definition \ref{path-def}. So, we define a \textit{fixing} shortest path algorithm. It is, the algorithm
is now allowed to do unforced flips in order to fix the graph but is still flipping the edges on the shortest path at a constraint violation.
We show maintaining constraint with a single flip at a step is not possible.

\begin{lemma} Given sufficient number of edges, there exists a sequence of edge additions such that a fixing shortest path algorithm
is forced to flip $2$ edges at a step. Consequently given sufficient number of edges, no fixing shortest path algorithm can maintain in-degree constraint $c = 2$ by doing
at most $1$ flip at a step.
\end{lemma}
\begin{proof}
  Assume the adversary has constructed eight trees with shortest path of length $1$, i.e. we have $t_{1, i} \in T^1$ for $i \in \{ 1, 2, \dots, 8\}$. Let $t_{1, i} \rightarrow t_{1, j}$ denote
  the edge addition that connects the roots of $t_{1, i}$ and $t_{1, j}$.
  Without loss of generality, we can assume such edges are oriented towards the trees with higher index.
  Then to force $2$ flips, adversary first gives the following sequence to the algorithm:
  $$ (t_{1, 1} \rightarrow t_{1, 2}), (t_{1, 3} \rightarrow t_{1, 4}), (t_{1, 5} \rightarrow t_{1, 6}),  (t_{1, 7} \rightarrow t_{1, 8}) .$$
  Then by giving the edges $(t_{1, 2} \rightarrow t_{1, 4})$ and $(t_{1, 6} \rightarrow t_{1, 8})$, adversary would have two trees with shortest path of length $2$ rooted at
  the roots of $t_{1, 4}$ and $t_{1, 8}$. Note that each addition causes a flip and renders the algorithm null.
  Finally by giving the edge $t_{1, 4} \rightarrow t_{1, 8}$, adversary causes 2 flips.
  So after constructing 8 $t_1$ trees, adversary can force $2$ flips by using $7$ more edges.

  What remains is to show how to construct $t_1$ trees.
  Clearly, constructing them like in Figure 2 would not work as the algorithm could break the tree with a single flip. For
  the given figure, the algorithm would just flip $(c, d)$ after adversary adds $(b, d)$.
  The idea is to construct long chains and connect their heads (the node with 0 out-degree). Observe that there is nothing algorithm
  can do to prevent adversary from constructing such chain without creating a $t_1$.
  The algorithm can pick an arbitrary node at somewhere on the chain and orient both sides to opposite directions but then the adversary can
  pick the longest side of the chain and proceed. Clearly, flipping an edge in-between would create a $t_1$.

  \begin{figure}[ht] \label{fig:chain}
    \centering
  \begin{tikzpicture}[->,>=stealth',auto,node distance=1.5cm,
    thick,main node/.style={circle,draw,font=\sffamily\bfseries}]

    \node[main node] (1) {};
    \node[main node] (2) [right of=1] {};
    \node[main node] (3) [right of=2] {};
    \node[main node] (4) [right of=3] {};
    \node[main node] (5) [right of=4] {};

    \path[every node/.style={font=\sffamily\small}]
      (1) edge node [right][above] {} (2)
      (2) edge node [right][above] {} (3)
      (3) edge node [right][above] {} (4)
      (4) edge node [right][above] {} (5);

  \end{tikzpicture}
  \caption{A chain with length 4.}
  \end{figure}
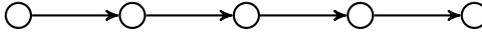

  We see that if we connect the heads of two chains, we would have a $t_1$ rooted at one of the heads.
  To break this $t_1$, the algorithm has to flip one of the edges that point to the heads of chains. But then, that would just create a new $t_1$. We see that the algorithm has to flip
  all the edges that lie on one of the two chains.
  Let $L$ be the length of the chains. The adversary first gives the input sequence that constructs $16$ distinct chains of length $L$ and algorithm
  orients them in some way. By the argument above, adversary has some combination of $t_1$ trees and
  chains of length at least $L/2$. If there are $8$ $t_1$ trees or more, adversary can force $2$ flips as described before. For all other cases,
  adversary can group the chains by $2$ (some chains might not be used), connect their heads and can have $8$ $t_1$ trees.

  We know explain what the sufficient number of edges mean. Assume that after inputting the algorithm, adversary gets $16$ chains of length $L/2$.
  After that point, adversary needs $8$ more edge additions to have $8$ $t_1$ trees. This simply says the chains must be at least of length 9, i.e. $L/2 \geq 9.$
  Thus, we see that we need at least $16 \cdot 18 + 8 + 7 = 303$ edges in total to construct $8$ $t_1$ trees and to force $2$ flips by connecting them in this pattern.
\end{proof}

\section{Online Bipartite $b$-Matching}
We now analyze the online bipartite $b$-matching.
First, we will give the definition of the problem.
Then, we will present the algorithm given in \cite{Gupta:2014:MAO:2634074.2634109} which achieves $O(1)$ competitive ratio by doing $O(n)$ recourse operations over $n$ arrivals.
As last, we will supply its analysis.

We are given a bipartite graph $G=(L, R, E)$ where nodes in $R$ are initially present and nodes in $L$ arrive one-by-one.
When a $v_L \in L$ arrives, it specifies a set of nodes in $R$ whose size is at least $1$. An edge is put between $v_L$ and all of the nodes in the set that $v_L$ specified.
Those nodes are the potential matches of $v_L$.
Among those potential matches, we have to select a single $v_R \in R$ and match $v_L$ to it by putting the edge $(v_L, v_R)$ to the matching set $M \subseteq E$.
While doing this over $n$ arrivals, we want to keep every $v_L$ matched and the degree of nodes in $R$ below some constant $b$.
More formally, $deg_G(v)$ is the degree of $v$ in $G$ and $deg_M(v)$ is the degree of $v$ in $M$. The latter one is also called $load(v)$ if $v \in R$.
Then for $n$ arrivals, we would like to maintain a subset of edges $M$ such that for every $v_L$, we have $deg_M(v_L) = 1$.  That is, every $v_L$ is matched with only one node in $R$.
While satisfying this constraint, we also want to have $\max_{v_R \in R}$ $load(v_R) \leq b$ which is defined as the cost of the problem.

We note that the bound on the competitive ratio given for the online edge orientation holds for this problem too.
If all the arriving nodes specify 2 possible matches, the upper bound and the lower bound on the competitive can be proven exactly as in the Lemma \ref{lwbnd}.
Again, we need recourse if we want to improve the competitive ratio.
For this problem, a recourse operation is defined as changing the match of a $v_L \in L$ from some $v_R \in R$ to a different $v_R' \in R$
where $v_R$ and $v_R'$ are among the potential matches of $v_L$. This is called as \textit{swap}.
The operation is defined so because every node in $L$ has to stay matched with some node in $R$ after its arrival.

In what follows, we present an algorithm that finds a solution with cost $CK$ for some $C \geq 2$ even though the given input allows a solution
with cost $K$. For the analysis, we set $C=2$ without loss of generality. We now give some preliminaries before presenting the authors algorithm.

\subsection{Preliminaries}
The algorithm runs on the residual graph $G^{res} = (L, R, A)$ with $A$ being a directed set of edges (arcs) that is constructed from the given bipartite graph $G$ and
with $L$ and $R$ being the same node sets in $G$. The set $A$ is constructed as follows: whenever we add the edge $(v_L, v_R)$ to $M$, we add
the arc $(v_R, v_L)$ (points to $v_L$) to $G^{res}$. For every other edge $(v_L', v_R') \in E \setminus M$, we add the arc $(v_L', v_R')$ to $G^{res}$.
Let $N(S)$ denote the set of nodes that have an incident edge whose other endpoint is in some node set $S$.
Similarly, let $N_M(S)$ denote the set of nodes that have an incident edge in $M$ whose other endpoint is in some node set $S$.
Respectively, those sets are called \textit{neighborhood of node set S} and \textit{M-neighborhood of node set S}.

A node $v_R \in R$ is called saturated if $2K$ $v_L \in  L$ are matched to it.
This is equivalent to $deg_M(v_R) = 2K$ and it means the out-degree of $v_R$ in $G^{res}$ is $2K$ as each of those matches will give rise to
an arc $(v_R, v_L)$ in $G^{res}$.
For any node $v$ in the given bipartite graph, $height(v)$ denotes the length of a shortest path in $G^{res}$ from $v$ to some unsaturated node.

\begin{lemma} \label{halls}
For any $S\subseteq L$, $N(S) \geq |S|/K $.
\end{lemma}
\begin{proof}
By Hall's theorem (\cite{Hall01011935}) and using the fact that the input allows a $K$ matching, we can write $N(S) \geq |S|$ for all $S \subseteq L$ when $K = 1$.
For an arbitrary $K \geq 2$, we just observe that if we duplicate every node in $N(S)$ for $K$ times we can treat the problem as if $K = 1$. So, if we apply the theorem
after the duplication we see that $K \cdot N(S) \geq |S|$.
\end{proof}

\begin{lemma} \label{sizeofsat}
For any $T\subseteq R$ whose all nodes are saturated, $N_M(T) = |T|\cdot 2K$
\end{lemma}
\begin{proof}
Each node in $T$ is matched to $2K$ nodes in $L$. Hence, the total number of matches of $T$ is simply $|T|\cdot 2K$.
\end{proof}

\subsection{Shortest Augmenting Path Algorithm}
Consider the arrival of a node $v_L$. When it arrives, it specifies its neighborhood set $N(v_L)\subseteq R$.
Algorithm then puts edges $(v_L, v_R)$ for all $v_R \in N(v_L)$ in $G^{res}$. If some nodes in $N(v_L)$ are unsaturated, the algorithm
picks one arbitrarily and puts the corresponding edge to $M$ (which also orients the edge accordingly). Now assume all the nodes in $N(v_L)$ are saturated.
Let $P$ be a shortest path to some unsaturated $y_k \in R$ from $v_L$.
From the construction of $G^{res}$, we know we can write the sequence of nodes on that path as,
$$ v_L=x_0, y_1, x_1, y_2, \dots, x_{k-1}, y_{k} $$ where each $x_i \in L$ and each $y_i \in R$.
We note that all the $y_i$ must be saturated except for $y_k$ by our assumption.

By looking at the sequence, we see that $y_1$ is a possible match for $x_0$ (as we have the arc $(x_0, y_1)$) but it is currently saturated.
Moreover, $x_1$ is currently matched to $y_1$ but it could be matched to $y_2$ if $y_2$ becomes unsaturated.
Looking at the end of the path we see that $x_{k-1}$ is currently matched to $y_{k-1}$ but it could be matched to $y_k$.
Since $y_k$ is unsaturated, we can match $x_{k-1}$ to $y_k$ and unsaturate $y_{k-1}$. Then, we can match $x_{k-2}$ to $y_{k-1}$ and unsaturate $y_{k-2}$ and so on.
Eventually, we can unsaturate $y_1$ and match $x_0$ to $y_1$.
Note that each $(x_i, y_i)$ pair in $M$ is replaced with $(x_{i-1}, y_i)$ for $i \in \{ 1, 2, \dots, k-1 \}$ and we added the new pair $(x_{k-1}, y_{k})$.

This is the main idea behind the algorithm. When a new node $v_L$  arrives, algorithm finds a shortest path from it to some unsaturated node and flips all the edges on that path.
This corresponds to making $(height(v_L) - 1)/2$ swaps.

\subsection{Analysis}
We will present the proof of the following theorem which gives an upper bound on the number of swaps over $n$ arrivals.

\begin{theorem} \label{btheo}
The total number of swaps done by the shortest augmenting path algorithm over $n$ arrivals is at most $2n$.
\end{theorem}
To prove the theorem, authors make use of a series of lemmas. We first present those. For the part below, $d_G(u, v)$ denotes the
length of a shortest directed path from $u$ to $v$ in a directed graph $G$, $|P|$ denotes the length of some path $P$ and
$P_{a \rightarrow b}$ denotes the subpath $a \rightarrow b$ that lies on $P$. The symbol $\circ$ denotes the concatenation of two directed paths.

\begin{lemma} \label{dhref}
Upon arrival of some $a \in L$, assume the algorithm flips the shortest path $P$ that ends at some unsaturated node $b$.
If $u$ and $v$ are two nodes on $P$ such that $u$ appears before $v$, then we have
$$ d_{H^{res}}(u, v) \geq d_{G^{res}}(u, v) $$
where $G^{res}$ is the residual graph before the flip of $P$ and $H^{res}$ is residual the graph after the flip of $P$.
\end{lemma}
\begin{proof}
First note that $H^{res}$ is just $G^{res}$ with the edges on $P$ flipped. This reversed path is denoted by $P^{rev}$ and
we write $P$ as $$ P = P_{a \rightarrow u} \circ P_{u \rightarrow v} \circ P_{v \rightarrow b} .$$
Let $Q$ be a shortest $u \rightarrow v$ path in $H^{res}$. Then there are two cases. Either $Q$ do not share any edge with $P^{rev}$ or it shares at least an edge with $P^{rev}$.

If $Q$ do not share any edge with $P^{rev}$, then we can write a path $P'$ from $a$ to $b$ in $G^{res}$ as
$$ P' = P_{a \rightarrow u} \circ Q \circ P_{v \rightarrow b} .$$
Since $P$ is a shortest $a \rightarrow b$ path in $G^{res}$ by our assumption, we have $|P'| \geq |P|$.
Moreover, $|P_{u \rightarrow v}| = d_{G^{res}}(u, v)$ as $P$ is a shortest path. \\
Thus, $ |Q| = d_{H^{res}}(u, v) \geq  d_{G^{res}}(u, v) = |P_{u \rightarrow v}|  .$

For the second case, we observe that if we concatenate the edges in $Q \setminus P^{rev}$, we get a path that starts at $u$ and ends at $v$.
Since this path is also present in $G^{res}$, its length is at least $d_{G^{res}}(u, v)$.
This implies $|Q| = d_{H^{res}}(u, v)$ is strictly greater than $d_{G^{res}}(u, v)$ in this case.

Hence, in general we have $d_{H^{res}}(u, v) \geq d_{G^{res}}(u, v).$
\end{proof}

\begin{lemma} \label{hnondec}
For any node x, height(x) is non-decreasing over arrivals.
\end{lemma}
\begin{proof}
We use the definitions from the previous lemma. Again, assume some $a \in L$ arrives and the algorithm flips a shortest $a \rightarrow b$ path $P$.
If $height(x) = h$ in $G^{res}$, then the lemma says that for any unsaturated node $y$ in $H^{res}$, the length of a shortest $x \rightarrow y$ path in $H^{res}$ is at least of length $h$.
Let $T$ be a shortest $x \rightarrow y$ path in $H^{res}$. We give a proof by showing $|T| \geq h$ in all possible cases.

First assume that $T$ does not share any edge with $P^{rev}$. If $y \neq b$, then this means flipping $P$ does not affect $T$. So, $T$ is also a shortest $x \rightarrow y$ path in $G^{res}$
and consequently, its length must be $h$. If $y = b$, then it is possible for $y$ to become saturated after the flip. If that happens, the shortest path from $x$ to some unsaturated node can not
possibly decrease as this would imply there exists a node $b'$ such that $x$ is closer to $b'$ than $b$ in $G^{res}$.
In short, $|T| \geq h$ if $T$ does not share any edge with $P^{rev}$.

Now assume $T$ shares at least an edge with $P^{rev}$. Let $u$ and $v$ be the first and the last node that appears in $T \cap P$ as we traverse the path $x \rightarrow y$.
First assume that $u$ appears before $v$ in $P$.
So if $Q$ denotes a shortest $u \rightarrow v$ path then, $$ T = T_{x \rightarrow u} \circ Q \circ T_{v \rightarrow y} .$$
Now we observe that $$ T' = T_{x \rightarrow u} \circ P_{u \rightarrow v} \circ T_{v \rightarrow y}$$  is a $x \rightarrow y$
path in $G^{res}$ and hence, $|T'| \geq h$. \\
The length of $T$ is given by
$$ |T| = |T_{x \rightarrow u}| + d_{H^{res}}(u, v) + |T_{v \rightarrow y}|$$
and the length of $T'$ is given by
$$ |T'| = |T_{x \rightarrow u}| + d_{G^{res}}(u, v) + |T_{v \rightarrow y}| .$$
Due to Lemma \ref{dhref}, we must have $|T| \geq |T'| \geq h$.

Now let $v$ appear before $u$ in $P$. First, we observe that
$ P_{a \rightarrow v} \circ T_{v \rightarrow y} $ is a path from $a$ to an unsaturated node. Hence, we must have
$$ |P_{a \rightarrow v}| + |T_{v \rightarrow y}| \geq |P| .$$
If we write $P$ as $ P = P_{a \rightarrow v} \circ P_{v \rightarrow u} \circ P_{u \rightarrow b} $, we can see that this inequality implies
$$ |T_{v \rightarrow y}| \geq |P_{v \rightarrow u}| + |P_{u \rightarrow b}| .$$
We further observe that $T_{x \rightarrow u} \circ P_{u \rightarrow b}$ is a path
from $x$ to an unsaturated node. Hence,
$$ |T_{x \rightarrow u}| + |P_{u \rightarrow b}| \geq h .$$
Finally, by summing the last two inequality we see that
$$ |T| \geq |T_{x \rightarrow u}| + |T_{v \rightarrow y}|  \geq h .$$
\end{proof}

\begin{lemma} \label{largeheight}
The number of nodes in $L$ with height at least $2h + 1$ is at most $|L|/2^h.$
\end{lemma}
\begin{proof}
The claim holds for $h = 0$. We analyze the case $h \geq 1$. Let $S_0 \subseteq L$ denote
the nodes with height at least $2h + 1 \geq 3$. Let $S_1$ be $N(S_0)$. Then, by Lemma \ref{halls} we have $|S_1| \geq |S_0|/K$.
Now observe that all the nodes in $S_1$ must be saturated. If they were not, some of the nodes in $S_0$ would have height 1
which would contradict with our assumption. Let $S_2$ be the set of nodes that are matched with nodes in $S_1$. Note that $S_2$ is
the set of nodes with height at least $2h-1$.
By Lemma \ref{sizeofsat}, we have $|S_2| = 2K \cdot |S_1| \geq 2|S_0|.$ By defining the sets $S_3, S_4, \dots, S_{2h}$ similary and repeating the argument up to
$S_{2h}$, we see that $2^h \cdot |S_0| \leq |S_{2h}|$. Since $|S_{2h}| \leq |L|$ as $S_{2h} \subseteq L$, we conclude that $2^h \cdot |S_0| \leq |L|$.
\end{proof}

\begin{lemma} \label{sumheight}
  Assume the algorithm flips the path $P$ upon arrival of some node $a \in L$.
  Let $height'(\cdot)$ denote the height function after the flip and let $height(\cdot)$ denote the height function before the flip.
  If the match of $y$ is changed from $x$ to $x'$ after the flip, then $height'(x') \geq height(x) + 2.$
\end{lemma}
\begin{proof}
  We again write the sequence of nodes in $P$ from $a$ to some unsaturated node $b$ as,
  $$ a = x_0, y_1, x_1, y_2, \dots, x_{k-1}, y_{k} = b $$
  where every $x_i \in L$ and every $y_i \in L$.
  We observe that $height(y_k) = 0$ and $height(x_i) = 2(k - i) -1$ as heights decrease by $2$ for each $x_i$ when traversing from $a$ to $b$.
  After flipping $P$, $y_i$ gets matched with $x_i$ and it was previously matched with $x_{i+1}$.
  Using the fact that $height(\cdot)$ is non-decreasing (Lemma \ref{hnondec}), we see
  \begin{align*}
  height'(x_i) &\geq height(x_i) \\
  & = height(x_{i+1}) + 2
  \end{align*}
\end{proof}

The last lemma shows that each swap increases the sum of heights at least by $2$. Further, we have a bound on the number nodes with height $2h + 1$ for any $h \geq 0.$
Authors make use of these two observations to prove Theorem \ref{btheo}. We restate it below.
\begin{proof}[Proof of Theorem \ref{btheo}]
Define the potential as
$$ \Phi = \sum_{x \in L}\frac{height(x) -1}{2} .$$
Clearly, it starts at $0$ as the first arriving nodes height is $1$.
A key observation about this potential is that each swap causes it to increase by at least 1 due to Lemma \ref{sumheight}. To see this more clearly, notice that
instead of summing over nodes in $L$, we could have summed over nodes in $R$ and then sum over their matches in $L$. We know from the Lemma \ref{sumheight} that
each swap causes a node in $R$ to match with a node in $L$ whose height is at least $2$ greater. Due to division by $2$, this will increase
potential at least by $1$. Since the potential starts at $0$ and since each swap causes it to increase at least by $1$, the value of this potential after $n$
arrivals is an upper bound on the number of swaps after $n$ arrivals. We compute the potential after $n$ arrivals and see that
\begin{align*}
\Phi &= \sum_{x \in L}\frac{height(x) -1}{2} \\
&  = \sum_{h=0}^{\infty} h \cdot (\text{number of nodes in L with height 2h + 1})  \\
& \leq |L| \cdot \sum_{h=0}^{\infty} \frac{h}{2^h}  \\
&  = 2|L| \\
&  = 2n,
\end{align*}
where the third line follows from Lemma \ref{largeheight}.
\end{proof}

As before, we note the tradeoff between the number of recourse operations and the constraint. For a $CK$ matching such that $C \geq 2$,
the result of Lemma \ref{largeheight} can be generalized as $|L|/C^h$ as the result of Lemma \ref{sizeofsat} would have changed to $|L| \cdot CK$.
Then, we would have $$ |L| \cdot \sum_{h=0}^{\infty} \frac{h}{C^h} = n \cdot \frac{C}{(C-1)^2}  $$ as an upper bound on the number of swaps
after $n$ arrivals.

\section{Conclusion}
Throughout the paper, the common theme was improving the competitive ratio by recourse. We presented
two problems in which no algorithm can do better than logarithmic competitive ratio when there is no recourse.
Moreover, we have seen even if we allow recourse to improve the competitive ratio, achieving the optimal offline solution can be costly.
In contrast to that, we observed approximating the optimal solution within a constant factor costs much less, e.g. the
algorithms we presented just needed a linear number of recourse operations in terms of arrivals.
We also did a brief analysis in an adversary against the algorithm setting for the shortest path algorithm and showed that it can be helpful to establish bounds.
In general, we have observed recourse can be a good approximation tool.
\newpage

\bibliography{Bibliography}
\bibliographystyle{ieeetr}
\end{document}